\documentclass[12pt]{article}
\usepackage{amsmath,amsfonts,amssymb,amsthm}
\usepackage{enumerate, multirow}
\usepackage{bm}
\usepackage{natbib} 
\usepackage{url,hyperref} 
\usepackage{caption, color, graphicx}
\usepackage{algorithm,algpseudocode}
\usepackage{longtable}

\theoremstyle{plain}
\newtheorem{theorem}{Theorem}
\newtheorem{assumption}{Assumption}

\newtheorem{lemma}{Lemma}

\theoremstyle{definition}

\newtheorem{remark}{Remark}
\newtheorem{example}{Example}

\newcommand{\indep}{\perp \!\!\! \perp}
\newcommand{\X}{{\bf X}}
\newcommand{\x}{{\bm{x}}}

\newcommand{\Y}{{\bf Y}}

\newcommand{\B}{{\bf B}}
\newcommand{\D}{{\bf D}}

\newcommand{\tX}{\tilde{\bf X}}
\newcommand{\tY}{\tilde{Y}}

\newcommand{\spc}{{\mathcal S}_{Y|\X}}

\newcommand{\cms}{{\mathcal S}_{E(Y|\X)}}

\newcommand{\R}{\mathbb{R}}
\newcommand{\mS}{\mathcal{S}}

\newcommand{\A}{{\bf A}}

\newcommand{\M}{{\bf M}}

\newcommand{\V}{{\bf V}}

\newcommand{\bP}{{\bf P}}

\newcommand{\var}{\text{Var}}
\newcommand{\Sig}{\bm{\Sigma}}

\newcommand{\cov}{\text{Cov}}

\newcommand{\bOmega}{\bm{\Omega}}

\newcommand{\I}{{\bf I}}

\DeclareMathOperator*{\argmin}{arg\,min}
\DeclareMathOperator{\E}{E}

\newcommand{\blind}{0}

\begin{document}

\def\spacingset#1{\renewcommand{\baselinestretch}%
{#1}\small\normalsize} \spacingset{1}


\if0\blind
{
  \title{\bf Detecting influential observations in single-index Fr\'echet regression}
  \author{Abdul-Nasah, Soale  \thanks{
    Corresponding Author: \textit{abdul-nasah.soale@case.edu}}\hspace{.2cm} \\
  Department of Mathematics, Applied Mathematics, and Statistics, \\
  Case Western Reserve University, Cleveland, OH, USA
 }
    \maketitle
} \fi

\if1\blind
{
  \bigskip
  \bigskip
  \bigskip
  \begin{center}
    {\LARGE\bf Title}
\end{center}
  \medskip
} \fi

\bigskip
\begin{abstract}
Regression with random data objects is becoming increasingly common in modern data analysis. Unfortunately, this novel regression method is not immune to the trouble caused by unusual observations. A metric Cook's distance extending the original Cook's distances of \cite{cook1977detection} to regression between metric-valued response objects and Euclidean predictors is proposed. The performance of the metric Cook's distance is demonstrated in regression across four different response spaces in an extensive experimental study. Two real data applications involving the analyses of distributions of COVID-19 transmission in the State of Texas and the analyses of the structural brain connectivity networks are provided to illustrate the utility of the proposed method in practice.
\end{abstract}

\noindent%
{\it Keywords:}  Metric Cook's Distance, Influential observations, Anomalous behavior, Fr\'echet regression
\vfill

\newpage
\spacingset{2} 

\section{Introduction\label{sec:1}}
\noindent Regression is one of the most powerful statistical methods used to investigate the relationship between quantities. Over the years, the methodology has been extended to incorporate quantities of various forms such as count measurements, categorical variables, and functional data. Recently, there has also been a growing interest in regression involving random objects such as shapes, networks, probability distributions, covariance matrices, spheres, and other complex objects. For some examples, see \citet{faraway2014regression, petersen2019frechet}, and \cite{zhang2023dimension}. 

While the flexibility of regression methods have allowed them to be adopted in many fields, the accuracy of most regression methods, especially those focused on estimating the conditional mean, suffer when unusual observations or outliers are present in the data. Such unusual occurrences are easy to spot from a scatter plot when dealing with scalar observations. For instance, looking at Figure \ref{fig:denswgt}, it is obvious that Subject 50 has an unusually large weight with a brain network that is structurally less dense. However, when dealing with networks, distributions, and other complex data types, it may be difficult to discern such unusual occurrence by simple visualizations because of the complex form of such data objects. For example, if we examine the actual networks for Subjects 1, 2, and 50 in Figure \ref{fig:net3}, it is hard to tell that the brain network of Subject 50 is less dense compared to the other two networks, or that the network for Subject 50 corresponds to a brain network of a person with a high weight. 

\begin{figure}[htb!]
        \centering
            \includegraphics[width=0.8\textwidth,]{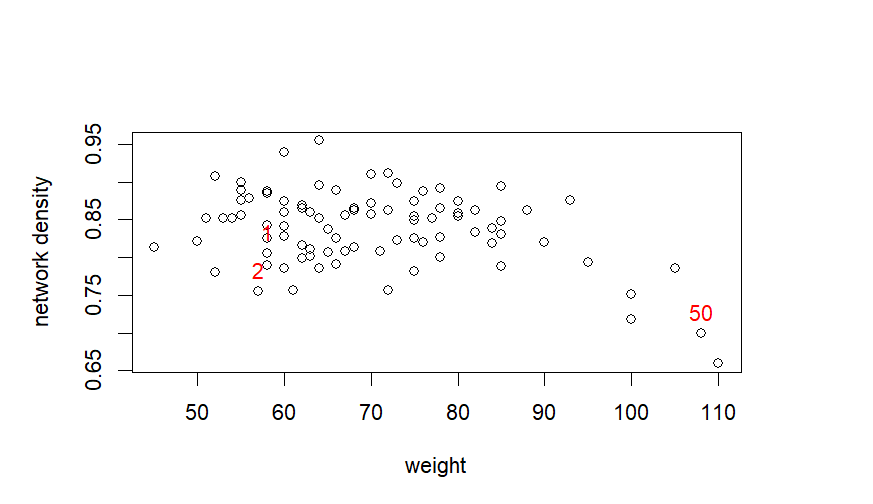}
            \caption{Example showing the scatter plot of structural brain network densities vs weight(kg) of the 88 healthy individuals in a study. The numbers in red represent specific subjects.}
            \label{fig:denswgt}
            \includegraphics[width=0.9\textwidth,]{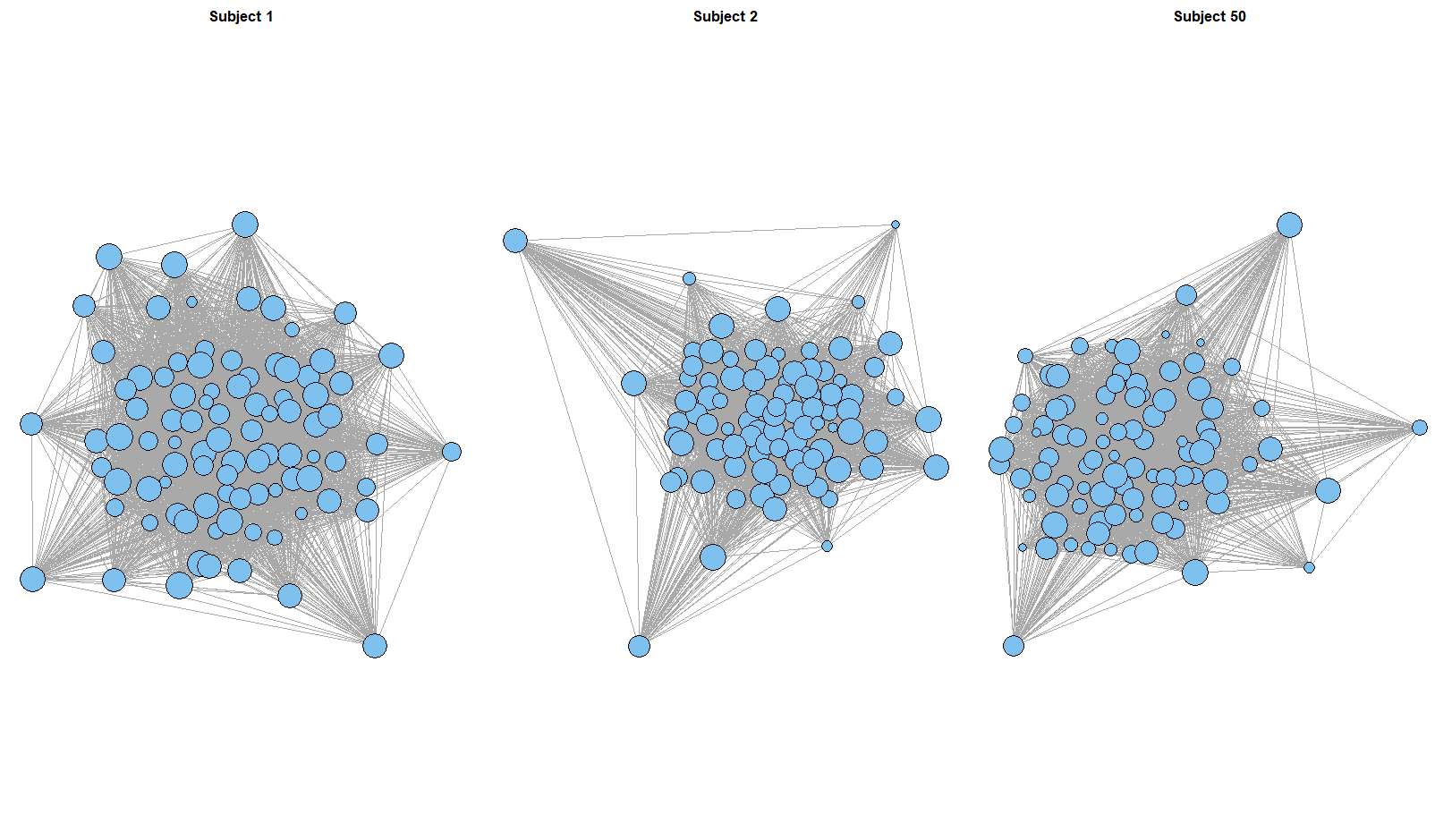}
            \vspace{-0.5in}
        \caption{Example showing structural brain connectivity networks of Subjects 1, 2, and 50 in the study. The blue circles represent 90 cortical regions of interests and the gray lines denote the links between the regions.}
        \label{fig:net3}
\end{figure}

It is important to note that outliers can be very informative, especially about the data collection process and may not necessarily cause problems in the statistical analysis. However, some outliers can distort the regression fit or cause the regression assumptions to be violated. Such outliers are referred to as {\it influential} and they are the focus of this study. In classical linear regression, Cook's distances of \cite{cook1977detection} is one of the popular techniques used to detect influential observations. For detecting influential observations in single index models and anomalous structure in functional data, several methods have been proposed. See \citet{prendergast2008trimming} and \cite{austin2024detection}, respectively. However, to the best of our knowledge there are no such procedures for detecting influential observations or anomalous behavior in regression between random response objects and Euclidean predictors. Given the growing need for Fr\'echet regression, it is important that we investigate how unusual observations can impact the accuracy of our analysis in this novel regression setting. 

Here, we will restrict our attention to mean regression between response objects that lie in bounded metric spaces and Euclidean predictors. Our proposal can be summarized as follows. First, find a Euclidean surrogate of the original response based on the pairwise distances between the response objects. Then fit an ordinary least squares (OLS) regression between the surrogate response and the original predictors. Repeat the procedure after leaving out one observation and calculate the corresponding Cook's distances following the procedure of \cite{cook1977detection} to determine observations that are influential. Therefore, we call our procedure the {\it metric Cook's distance}. We will demonstrate this idea in regression settings involving responses in four different metric spaces. The goal is to first detect observations that are influential and then investigate those observations by examining how trimming them from the data affect the central mean space basis estimates. It is worth noting that the metric Cook's distance also comes in handy in studies where the investigator's goal is to detect anomalies in the data.

The rest of the paper is organized as follows. Section~\ref{sec:2} gives a brief background to regression between random response objects and Euclidean predictors, the model under consideration, and the estimation of the metric Cook's distances. In Section~\ref{sec:3}, we demonstrate the performance of the metric Cook's distance through extensive experimental studies. Two real applications of the metric Cook's distance involving the analyses of COVID-19 transmissions in the State of Texas and the analyses of structural brain connectivity networks of healthy individuals are given in Section~\ref{sec:4}. The concluding remarks are given in Section~\ref{sec:5}. All proofs and technical details are relegated to the appendix.

\section{Brief Background and Methodology \label{sec:2}}
To fix the idea, let $Y$ denote a random response in the metric space $(\Omega_Y,d)$, where the metric $d:\Omega_Y \times \Omega_Y \to \R$. Let $\X$ be a $p$-dimensional predictor. Next, assume a random object $(\X, Y) \in \R^p \times \Omega_Y$ has a joint distribution $F_{XY}$, where $\X$ and $Y$ have marginal distributions $F_X$ and $F_Y$, respectively. If the conditional distributions $F_{Y|X}$ and $F_{X|Y}$ exist and are well-defined, then \cite{petersen2019frechet} defined the conditional Fr\'echet mean of $Y$ given $\X$ as
\begin{align}
    m_{\bigoplus}(\x) = \argmin_{\omega\in \Omega_Y} \E\big( d^2(Y,\omega)|\X=\x \big).
    \label{fre_reg}
\end{align}

Notice that model (\ref{fre_reg}) is analogous to $\E(Y|\X=\x)$, where $Y\in\R$. Thus, Fr\'echet regression shares some of the problems encountered in classical regression  associated with $\X$, such as large $p$ and the effect of outliers. The former motivated the development of dimension reduction for Fr\'echet regression, see \citet{dong2022frechet, zhang2023dimension}, and \cite{soale2023data}. In this study, we focus on the impact of unusual observations on single-index Fr\'echet regression models, which we will present in the Sections that follow. 

\subsection{Model and Dimension Reduction Subspace \label{subsec:model}}
We start by assuming that the relationship between $Y\in\bOmega_Y$ and $\X\in\R^p, p\geq 1$, is given by the single-index model:
\begin{align}
    Y = f(\beta^\top\X, \epsilon),
    \label{sim}
\end{align}
where $f$ is some unknown link function, $\beta\in\R^p$ with $\lVert \beta \rVert_2 = 1$, and $\lVert .\rVert$ denotes the $\ell_2$ norm. $\epsilon$ is the noise with $\E(\epsilon) = 0$ and $\epsilon \indep \X$, where $\indep$ means statistical independence. 

In Model \ref{sim}, $Y$ depends on $\X$ only through $\beta^\top\X$. Therefore, our goal is to find $\beta$ without imposing stringent assumptions on $f$, which may not even have a conceivable form. Since our focus is on mean regression, we want a $\beta$ that satisfies 
\begin{align}
    Y \indep \E(Y|\X) | \beta^\top\X, \text{ or equivalently, } \E(Y|\X) = \E(Y|\beta^\top\X).
    \label{cms}
\end{align}
Clearly from (\ref{cms}), we see that $\beta$ is not unique as any scalar multiple of $\beta$ also satisfies (\ref{cms}). Luckily under some mild conditions, the intersection of all such $\beta$ is unique, and if it exists, it is called the central mean space for the regression of $Y$ on $\X$, denoted as $\cms = \mathrm{span}(\beta)$. See \cite{cook1996graphics} and \cite{yin2008successive} for the existence and uniqueness of the central mean space. 

In the sufficient dimension reduction literature, several estimators of the central mean space have been proposed. One of the most popular methods is the ordinary least squares (OLS; \cite{li1989regression}). The OLS method is a powerful central mean space estimator in single index models with monotone link functions and a key reason for its popularity is its ease of implementation. However, OLS is susceptible to unusual observations, which also makes it a useful tool for detecting influential observations.

\subsection{Surrogate Ordinary Least Squares\label{subsec:sOLS}}
We consider a general response $Y\in (\bOmega_Y,d)$, which includes both Euclidean and non-Euclidean metric spaces. In non-Euclidean spaces, the usual vector space operations such as inner products are not always guaranteed, which makes it infeasible to implement classical OLS. However, the good news is that in every metric space we can at least compute distances between the elements of $\bOmega_Y$, if the set is bounded. Therefore, to accommodate regression with responses beyond the Euclidean space, we propose a {\it surrogate ordinary least squares}. The surrogate OLS will utilize a surrogate response based on the pairwise distances between the response objects. 

In the existing literature, there are extensions of OLS that accommodate responses in non-Euclidean spaces. These include the Fr\'echet OLS of \cite{zhang2023dimension}, which we will call the kernel Fr\'echet OLS (FOLS) and the surrogate-assisted OLS (sa-OLS) \cite{soale2023data}. However, we propose an alternative surrogate OLS because FOLS requires that the response metric be continuously embeddable in the Hilbert space. This condition restricts our choice of kernels and also involves choosing optimal tuning parameters for the kernel estimation. While sa-OLS does not require the continuous embedding condition, it requires several random projections of the distance matrix, which requires choosing a large number of projection vectors and can be computationally intensive.

Following in the spirit of \cite{faraway2014regression}, we propose a surrogate OLS via multidimensional scaling. The Euclidean surrogate of $Y\in (\bOmega_Y,d)$ obtained via multidimensional scaling preserves the distances between the original response objects. Like the classical OLS and its extensions, our surrogate OLS method also makes the following standard assumptions.  

\begin{assumption}\label{monotonicity}
We assume the function $f$ is monotone.    
\end{assumption}

\begin{assumption}\label{LCM}
We assume that $\E(\X|\beta^\top\X)$ is a linear function of $\beta^\top\X$. 
\end{assumption}

\noindent Assumption \ref{monotonicity} is needed because the OLS method is known to perform poorly when the link function is symmetric. Assumption \ref{LCM} is called the linear conditional mean (LCM) assumption, which is famous in the sufficient dimension reduction literature. The LCM is shown to be satisfied when the distribution of $\X$ is elliptical or when $p$ is very large, see \cite{eaton1986characterization} and \cite{li2018sufficient} for details. The upcoming lemma guarantees the recovery of an unbiased estimate of $\cms$ using regression based on the surrogate response.

\begin{lemma}\label{lemma:surr_sdr}
Let $\tY$ be a random copy of $Y \in (\bOmega_Y,d)$. Define the surrogate $S_Y = \phi\big(d(Y,\tY)\big)$, for some measurable function $\phi: d(Y,\tY) \to \R$. Then, $\mathcal{S}_{S_Y|\X} \subseteq \spc$. 
\end{lemma}

\noindent Lemma \ref{lemma:surr_sdr} follows directly from Proposition 1 in \cite{soale2023data} and it implies that we can use the regression of $S_Y$ versus $\X$ to estimate $\cms$. 

\begin{theorem}\label{thrm:surr_ols}
Suppose Assumptions \ref{monotonicity} and \ref{LCM} are satisfied and that $\Sig = \var(\X)$ is invertible. Then
    \begin{align}
    \Sig^{-1}\Sig_{X S_Y} \in \cms, \text{ where } \Sig_{X S_Y} = \cov(\X, S_Y).
    \label{ols}
\end{align}
\end{theorem}

\subsection{Metric Cook's Distance Estimation \label{subsec:mCD}}
Given a random sample $(\x_1, y_1), \ldots, (\x_n, y_n)$ of $(\X, Y) \in \R^p \times \Omega_Y$. We compute the $n \times n$ pairwise distance matrix $\D_Y$ such that $\D_{Y_{ij}} = d(y_i,y_j)$, for $i,j=1,\ldots,n$. Note that the metric $d$ depends on the metric space, and for a given metric space, there may be several choices for $d$. For instance, if $Y \in \R^q, q \geq 1$, we use $\D_{Y_{ij}} = \lVert y_i-y_j\rVert_2$. However, another investigator may choose a different distance applicable in the Euclidean space, say, the Manhattan distance. Similarly, if $\bOmega_Y$ is the collection of probability distributions, we follow the lead of \cite{petersen2019frechet} and use the Wasserstein metric, but we believe other metrics in the space such as the Hellinger distance will also work. In effect, we surmise that the choice of distance metric depends on the space of the response objects and the decision is left to the discretion of the practitioner where there are multiple choices for $d$, at least for now. 

Regardless of the response space and metric choice, $\D_Y \in \R^{n\times n}$ is expected to be a symmetric matrix whose diagonal entries are all zeros. Thus, $\D_Y$ is a dependent matrix. To obtain surrogate Euclidean vectors with the same distance matrix as $\D_Y$, we apply metric multidimensional scaling (MDS). MDS is a common dimension reduction technique used to estimate factor scores from the singular value decomposition of the double-centered squared distances. See \cite{kruskal1964multidimensional, kruskal1978multidimensional} for more on MDS. 

The MDS procedure is as follows. First, find the double-centered squared distance matrix
\begin{align}
    \M = -\cfrac{1}{2} \bm Q_n \D^2_Y \bm Q_n,
    \label{mmds}
\end{align}
where $\bm Q_n = \I_n - n^{-1}\bm J_n$, $\I_n$ is an $n\times n$ identity matrix, and $\bm J_n$ is an $n\times n$ matrix of all 1's. Thus, $\M$ is positive semi-definite with a spectral decomposition $\M = \V\bm \Lambda \V^\top$, where $\bm\Lambda = diag\{\lambda_1, \ldots, \lambda_k\}$ with $\lambda_1 \geq \lambda_2 \ldots \geq \lambda_k > 0$ and $\V^\top\V = \I_k$. We then proceed to find the first $k$ factor scores $S = \V\Lambda^{1/2} \in \R^{n\times k}$. Here, we retain only the first factor score as our surrogate response. Ideally, we would have to determine the number of scores to retain. However, since we focus only on single-index models, we choose the first score, which corresponds to the eigenvector of the leading eigenvalue of $\M$. 

Once we estimate $S_Y$, we proceed to find the Cook's distances for the regression between $S_Y$ and $\X$. The calculation of the Cook's distances involves the intercept estimate, which is incorporated into (\ref{ols}) by augmenting the $\X$ matrix with a vector of 1's. Let $\tX = [1, \X]$ and the corresponding parameter estimate $\tilde\beta = \tilde\Sig^{-1}\Sig_{\tilde X S_Y}$, where $\tilde\Sig=\var(\tilde\X)$. Denote $\tilde\beta^{(-i)}$ as the OLS estimate when the $i$th observation is left out. The $i$th Cook's distance is defined as
\begin{align}
    \delta_i = \cfrac{(\tilde\beta^{(-i)} - \tilde\beta)^\top(\tX^\top\tX)(\tilde\beta^{(-i)} - \tilde\beta)}{(p+1)s^2},
\end{align}
where $s^2 = \cfrac{1}{n-p-1}\displaystyle\sum_{i=1}^n (S_{Y_i} - \tilde\beta^\top\tX_i)^2$. Large values of $\delta_i$ indicate that the $i$th observation is likely influential and warrants further investigation. While there is no standard threshold for classifying an observation as influential as several cutoffs have been suggested in the literature, we use $4/(n-p-1)$ for the sake of visualization. 

\begin{remark}
Unlike the original Cook's distance where the response after deleting the $i$th observation is simply $Y^{(-i)}$, here $S_{Y^{(-i)}}$ is based on $\D_{Y^{(-i)}}$, which is not necessarily the same as deleting the $i$th observation from $S_Y$.  
\end{remark}

\noindent The estimation is summarized in the algorithm that follows.

\begin{algorithm}[htb]
\caption{Metric Cook's Distance}
\begin{algorithmic}[1]
   \State  Input: Predictor $\X$ as $(n \times p)$ matrix and response $Y_n = \{y_1,\ldots,y_n\}$ as list
   \State Compute the $n\times n$ distance matrix $\D_Y$, where $\D_{Y_{ij}} \gets d(y_i, y_j), \ \forall i,j = 1\ldots, n$
   \State Compute the MDS factor scores based on $\D_Y$ and set the first score as $S_Y$
   \State Compute the OLS estimate $\tilde\beta$ from the regression between $S_Y$ and $\X$
   \State Repeat Steps 2-4 after deleting the $i$th observation and compute the $i$th Cook's distance
\end{algorithmic}
\end{algorithm}

\section{Experimental Studies\label{sec:3}}
\noindent In this section, we will demonstrate the performance of the metric Cook's distance on synthetic data in regression with responses across four different metric spaces. For each response space, we will examine the estimation accuracy and the effect of trimming the influential observations on the kernel Fr\'echet OLS (FOLS) and the surrogate-assisted OLS (sa-OLS) basis estimates, where applicable. The accuracy of each method is measured using 
\begin{align}
	\Delta = \lVert \bP_{\beta} - \bP_{\hat\beta} \rVert_F,
\end{align}  
where $\bP_\A = \A(\A^\top \A)^{-1}\A^\top$ and $\lVert.\rVert_F$ is the matrix Frobenius norm. $\Delta$ is a popular criterion for evaluating methods in the SDR literature. Smaller values of $\Delta$ indicate better performance. 

\subsection{Metric Cook's Distance in Different Metric Space}
While large samples are more likely to contain outliers, the effect of each outlying observation on the estimates is less when the sample is large. Therefore, we focus our analysis on small and moderate sample sizes. For the upcoming examples, we fix $(n, p) = (100, 5)$ and then generate the predictor as $\X=(X_1,\ldots,X_5) \sim t_{20}(\bm 0, \I_5)$. The multivariate $t$ distribution satisfies the linearity condition in Assumption \ref{LCM}. Also, because of the heavy tails of the $t$ distribution, we expect $\X$ to contain some outliers, which may not necessarily be obvious by visualization. 

Let the true basis vector be $\beta^\top=(1,1,0,0,0)/\sqrt{2}$. Because we are dealing with subspace estimation rather than parameter estimation, we think of our estimates in terms of loadings. This means we expect the first two coefficients of the $\cms$ basis estimate to be non-zero and equal in proportion and the remaining coefficients to be close to zero. 

\subsubsection{Euclidean Response}
\noindent The original Cook's distance method of \cite{cook1977detection} readily applies in regression with $Y\in\R$ and yields the same results as the metric Cook's distance as $S_Y$ is simply a scaled $Y$ in this case. Thus, we focus on the high dimensional response $\Y \in \R^q, \ q \geq 2$, where the metric Cook's distance is more useful.

\begin{example}[Reduced-rank vector response] \quad \\
Generate a response in $\R^2$ as $\bm Y_i = \sin(\pi/2 + \B^\top\X_i) + \bm\epsilon_i$, for $i=1,\ldots, n$, where 
$\B^\top = (\beta,-2\beta)^\top $ and noise $\bm\epsilon_i \overset{i.i.d.}{\sim} N(\bm 0, \Sig_\epsilon)$ with $\Sig_\epsilon = \begin{pmatrix} 1 & 0.5 \\ 0.5 & 1 \end{pmatrix}$. Here $\B$ has rank 1 and $\beta$ is the true basis of $\cms$.
\end{example}  

The plot of the metric Cook's distance are given in exhibit (A) of Figure \ref{fig:mcd_euc}. 

\begin{figure}[htb!]
    \centering
    \includegraphics[width=0.78\textwidth,]{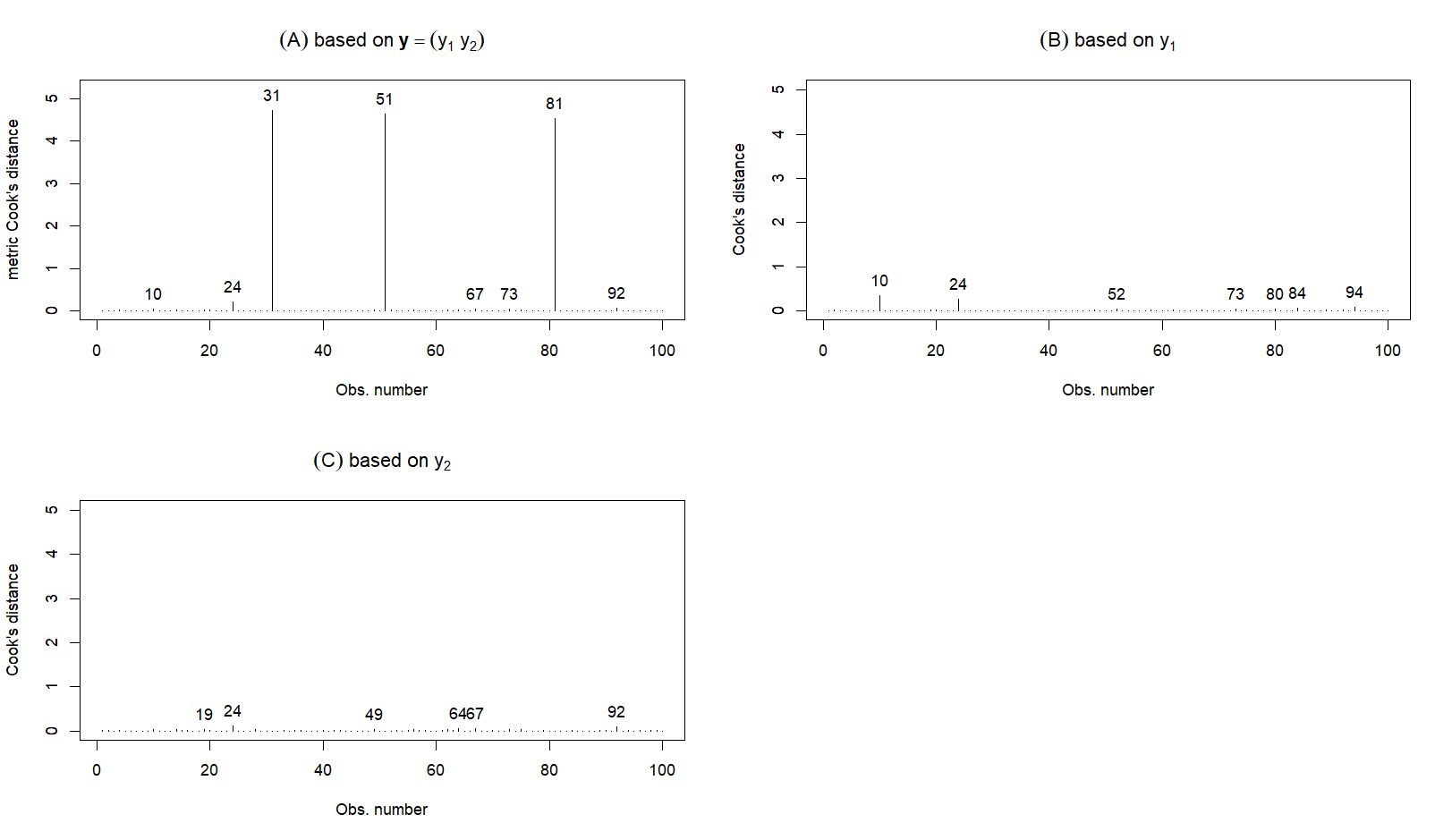}
    \caption{Cook's distances in (A) are based on OLS regression with $S^Y$ as response while those for exhibits (B) and (C) are based on separate regressions with $y_1$ and $y_2$, respectively.}
    \label{fig:mcd_euc}
    \includegraphics[width=0.75\textwidth,]{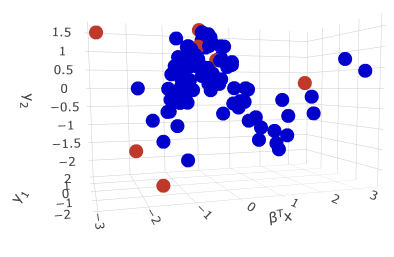}
    \caption{Scatter plot of the bivariate response $\bm y$ versus the sufficient direction $\beta^\top\x$. Influential observations based on the metric Cook's distance are colored red.}
    \label{fig:eup_resp}
\end{figure}

Exhibits (B) and (C) represent the original Cook's distance from the classical OLS regression. Note that implementing the classical OLS with multiple responses is equivalent to doing so for each response and stacking the coefficients together in columns. Thus, the classical OLS ignores the dependence information between the responses. This however, is not a problem when using the metric Cook's distance. We also see that the metric Cook's distances capture most of the influential observations flagged by the marginal Cook's distances in (B) and (C). Interestingly, the most influential observations based on the joint response, i.e., 31, 51, and 81 in exhibit (A) were not even captured by the marginal Cook's distances. This is one of the disadvantages of using marginal Cook's distances for multiple responses. Another disadvantage of using the marginal Cook's distances for $Y \in \R^q, \ q > 1$, is that we have to examine $q$ Cook's distance plots for each response separately, which can be tedious for large $q$. The 3D scatter plot of $\bm y=(y_1,y_2)$ vs the sufficient predictor $\beta^\top\x$ is given in Figure \ref{fig:eup_resp}. The influential observations indicated in red are based on the metric Cook's distance. 

To further investigate the effect of the influential observations, we find the basis estimates with and without them. The result is presented in the Table that follows. 

\begin{table}[htb]
\centering
\caption{* denotes estimates without influential observations}
\label{ex:euc}
\begin{tabular}{lccccc|c}
\hline
 & $x_1$ & $x_2$ & $x_3$ & $x_4$ & $x_5$ & $\Delta$ \\ \hline    
sa-OLS  & -0.6113 & -0.5810 & 0.1662  & -0.0724 & 0.3617 & 0.6125 \\
sa-OLS* & -0.5149 & -0.7949 & -0.0996 & -0.2098 & 0.2345 & 0.5428\\
\hline
\end{tabular}
\end{table}
\noindent Looking at the coefficients in Table \ref{ex:euc}, the sa-OLS appears to be doing a good job at recovering the basis of the central mean space with a slight improvement when we omit the influential observations (inf. obs). We did not consider FOLS here because it was not implemented with Euclidean responses in \cite{zhang2023dimension}.

\subsubsection{Distribution as Response}
\noindent For this class of response objects, we use the Wasserstein metric to find the pairwise distances. The Wasserstein metric for univariate distributions can be expressed in terms of the $\ell_p$ norm of the quantiles of the input distributions. The $k$th Wasserstein is given by 
\begin{align}
 W_k(y,y') = \lVert F_{y}^{-1} - F_{y'}^{-1}\rVert_k = \left(\int_0^1 \big\lvert F_{y}^{-1}(s)  - F_{y'}^{-1}(s)\big\rvert^k ds \right )^{1/k},
\label{wasserstein}
\end{align}
where $F_{y}^{-1}$ and $F_{y'}^{-1}$ are the respective quantile functions of $y$ and $y'$. For the surrogate response, we use the leading MDS score based on $W_1(.)$ distances as it is more robust to outliers compared to the $W_2(.)$ metric used in \cite{petersen2019frechet} and other previous studies.

\begin{example} \quad \\
Generate each response as a mixture of normal distributions as follows: 
\begin{align*}
    Y_i \in \R^{100} \overset{i.i.d.}{\sim} 0.6N(\beta^\top\X_i, 1) + 0.4N(-\beta^\top\X_i, 2), \text{ for } i= 1,\ldots,n.
\end{align*}
\end{example}

\noindent The metric Cook's distances are given in Figure \ref{fig:MCD_dist} and the 3D plot of the densities versus the sufficient predictor ($\beta^\top\x$) is given in Figure \ref{fig:dist_resp}. The densities in red denote the influential or anomalous distributions. From Figure \ref{fig:dist_resp}, it appears that distributions associated with extreme values of $\beta^\top\x$ and distributions with more complex shapes are flagged as influential. Thus, it is important when using a metric other than the Wasserstein, the investigator chooses a metric that can capture even slight differences in shape.

\begin{figure}[htb!]
    \centering
    \includegraphics[width=0.75\textwidth, height=0.4\textwidth]{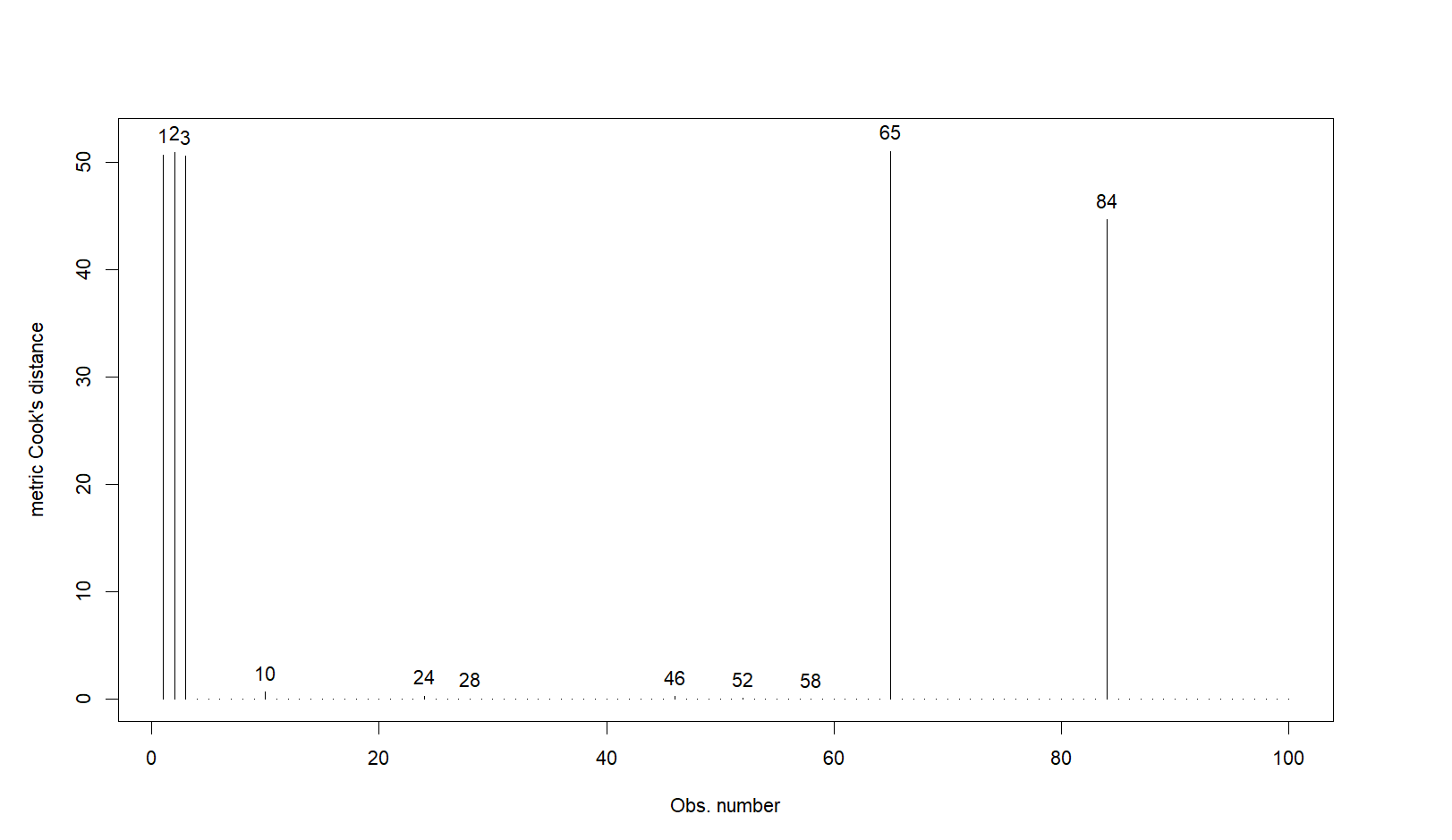}
    \caption{Metric Cook's distances for distributional response regression}
    \label{fig:MCD_dist}
     \includegraphics[width=0.75\textwidth, height=0.5\textwidth]{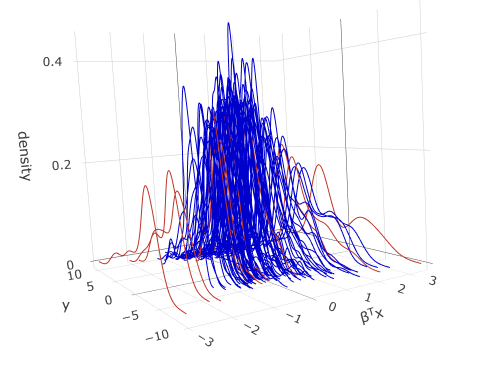}
    \caption{3D plot of distributions vs sufficient predictor $\beta^\top\x$. The red lines indicate the influential observations based on the metric Cook's distance.}
    \label{fig:dist_resp}
\end{figure}

\noindent Both the FOLS and sa-OLS estimators are applicable here and their results are provided in Table \ref{tbl:dist}. Both methods show significant improvement after omitting the influential observations. 

\begin{table}[htb!]
\centering
\caption{* denotes estimates without influential observations}
\label{tbl:dist}
\begin{tabular}{lccccc|c}
\hline
 & $x_1$ & $x_2$ & $x_3$ & $x_4$ & $x_5$ & $\Delta$ \\ \hline    
FOLS  & 0.7752  & 0.5521  & -0.1268 & 0.0155  & 0.0135  & 0.2995 \\
FOLS* & 0.8237  & 0.7880  & -0.0164 & -0.0087 & -0.0289 & 0.0529 \\
sa-OLS  & -0.7673 & -0.5326 & 0.1395  & -0.0634 & 0.0151  & 0.3382 \\
sa-OLS* & 0.8284  & 0.7707  & -0.0325 & 0.0203  & -0.0466 & 0.0909 \\ \hline
\end{tabular}
\end{table}

\subsubsection{Network as response}
Here, we consider our response as a network or graph, denote as $G=(V,E)$, where $V$ is the set of vertices or nodes and $E$ is the set of edges or links. The networks are allowed to have weighted and unweighted edges. We employ two distance measures: the {\it centrality distance (CD)} of \cite{roy2014modeling} and the {\it graph diffusion distance (DD)} of \cite{hammond2013graph}. The centrality distance is based on unweighted edges and is given by
    \begin{align}
        d_{CD}(G_1,G_2) = \displaystyle\sum_{v\in V} |C(G_1,v) - C(G_2,v)|.
        \label{cent_dist}
    \end{align}
The diffusion distance on the other hand incorporate the weight of the edges and is given by 
\begin{align}
    d_{DD,t}(G_1,G_2) = \max_t \big(\lVert e^{-t\L_{G_1}} - e^{-t\L_{G_2}}\rVert_F^2\big)^{1/2},
\label{diff_dist}
\end{align}
where for graph $G_i$, the Laplacians $\L_{G_i} = \D_{G_i} - \A_{G_i}$, $\D_{G_i}$ is the diagonal matrix of degrees, and 
$\lVert . \rVert_F$ denotes the matrix Frobenius norm. $t$ is the time which we take to be 1 because we are only dealing with static networks. 

\begin{example}
Generate the random network response $y_i$ based on an Erd\"os–R\'enyi model with probability $p = plogis\big(\sin(\beta^\top\X_i)\big), \ i=1,\ldots,n$, 
where $plogis(x) = \cfrac{1}{1+e^{-x}}$ and $plogis(.)$ is the cumulative distribution function (CDF) of the logistic distribution.    
\end{example}

\noindent In this example, we use the surrogate response based on the graph diffusion distances to estimate the metric Cook's distances. The metric Cook's distances are given in Figure \ref{fig:MCD_network}. 

\begin{figure}[htb!]
    \centering
    \includegraphics[width=0.85\textwidth,]{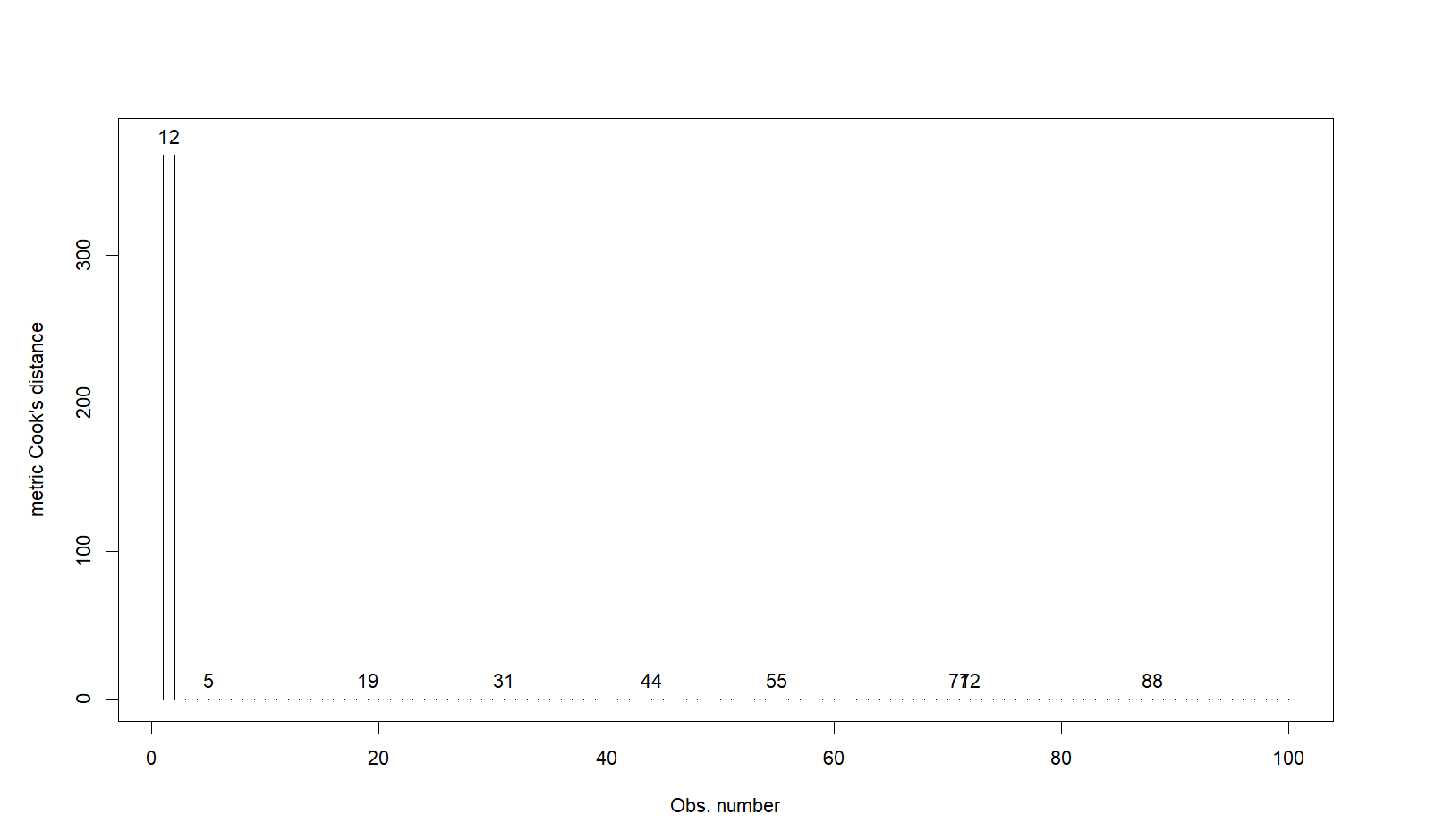}
    \caption{Metric Cook's distances for network response regression.}
    \label{fig:MCD_network}
\end{figure}

\noindent Since we are unable to plot all the 100 networks due to space constraints, we illustrate a sample of the normal and influential/anomalous networks in Figure \ref{fig:network_resp}. Although, it is hard to tell which network is anomalous, it appears the networks with fewer or more than the expected number of edges are flagged as influential or anomalous. 

\begin{figure}[htb!]
    \centering
    \includegraphics[width=0.95\textwidth, height=0.6\textwidth]{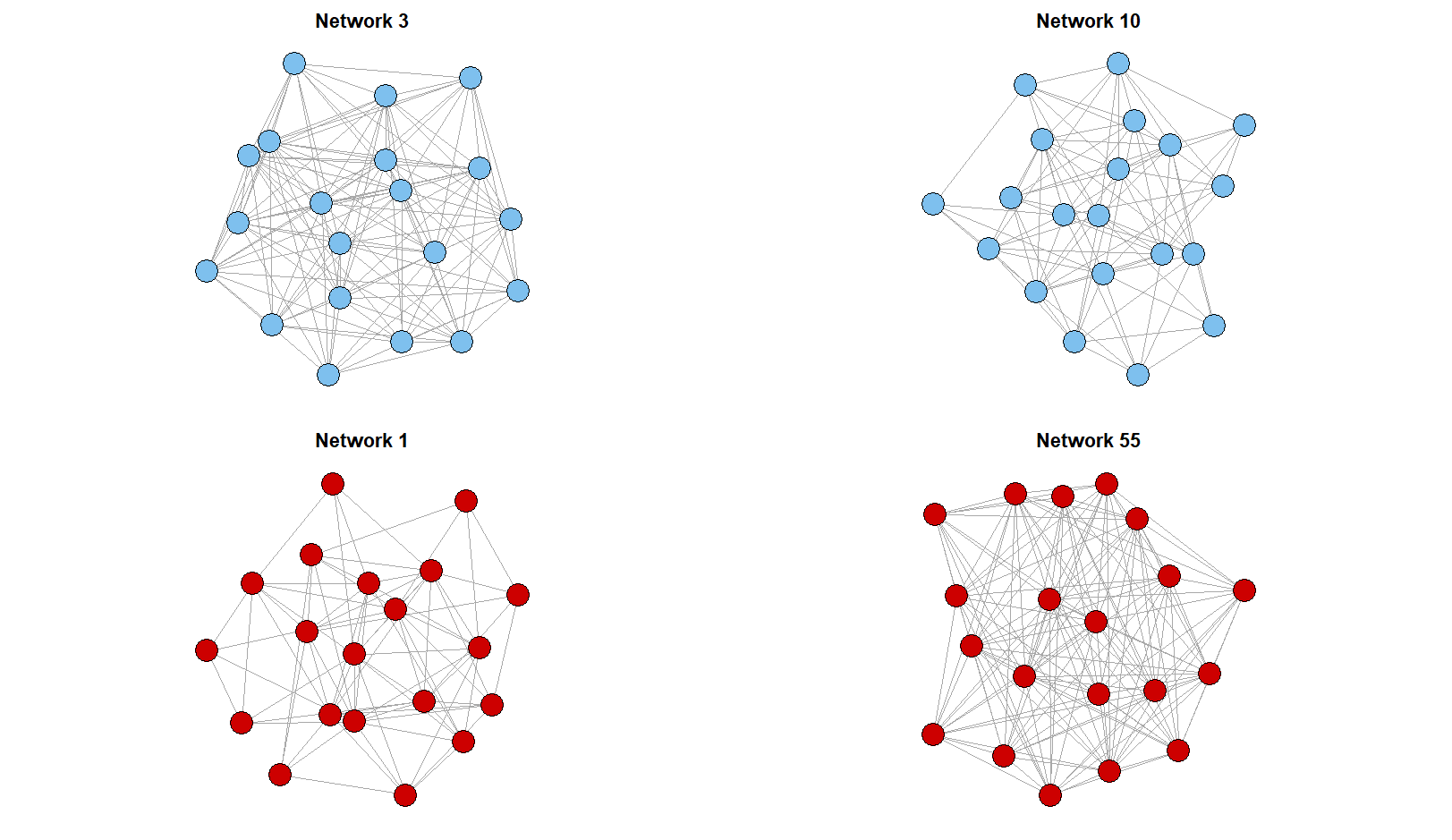}
    \caption{Four sample networks out of the 100 networks. The networks with red nodes indicate the influential/anomalous networks based on the metric Cook's distance.}
    \label{fig:network_resp}
\end{figure}

In Table \ref{ex:net}, we provide the sa-OLS estimates based on both the centrality distances (CD) and diffusion distances (dd). Both estimators show significant improvement after omitting the influential observations. The estimators also appear to be very similar as expected since the networks are unweighted. We did not consider FOLS as it was not implemented with network responses in \cite{zhang2023dimension}.

\begin{table}[htb!]
\centering
\caption{* denotes estimates without influential observations}
\label{ex:net}
\begin{tabular}{lccccc|c}
\hline
 & $x_1$ & $x_2$ & $x_3$ & $x_4$ & $x_5$ & $\Delta$ \\ \hline   
sa-OLS (cd)  &  0.7049 & 0.7442 & -0.0042 & -0.0022 & -0.1193 & 0.1680 \\
sa-OLS (cd) * & 0.8422 & 0.8164 & -0.0246 & -0.0101 & -0.0897 & 0.1146 \\ 
sa-OLS (dd)  & -0.7043 & -0.7372 & 0.0163 & -0.0106 & 0.0835 & 0.1228 \\
sa-OLS* (dd) & 0.8515 & 0.8075 & -0.0230 & 0.0129 & -0.0635 & 0.0908  \\
\hline
\end{tabular}
\end{table}

\subsubsection{Functional response}
\noindent For this response space, we follow \cite{soale2023data} and use the discrete Fourier distances. The interested reader may refer to \cite{faraway2014regression} to see how Fr\'echet distance can be used. 

\begin{example} \quad \\
We generate each response based on 30 random time points $t\in\R^{30} \overset{i.i.d.}{\sim} Unif(0, 10)$. Let $\alpha(t) = 2\sin(\pi + \pi t/5)$ be the time-varying intercept. The responses are generated as $Y_i(t) = \alpha(t) + 2\sin\left(\pi t/2 + \beta_1^\top\X_{i}\right) + \epsilon_{i}(t)$ for $i=1,\ldots,n$,
where $\epsilon_{i}(t) \overset{i.i.d.}{\sim} N(0,1)$. 
\end{example}

\noindent The metric Cook's distances are provided in Figure \ref{fig:MCD_func}. 

\begin{figure}[htb!]
    \centering
    \includegraphics[width=0.85\textwidth, height=0.5\textwidth]{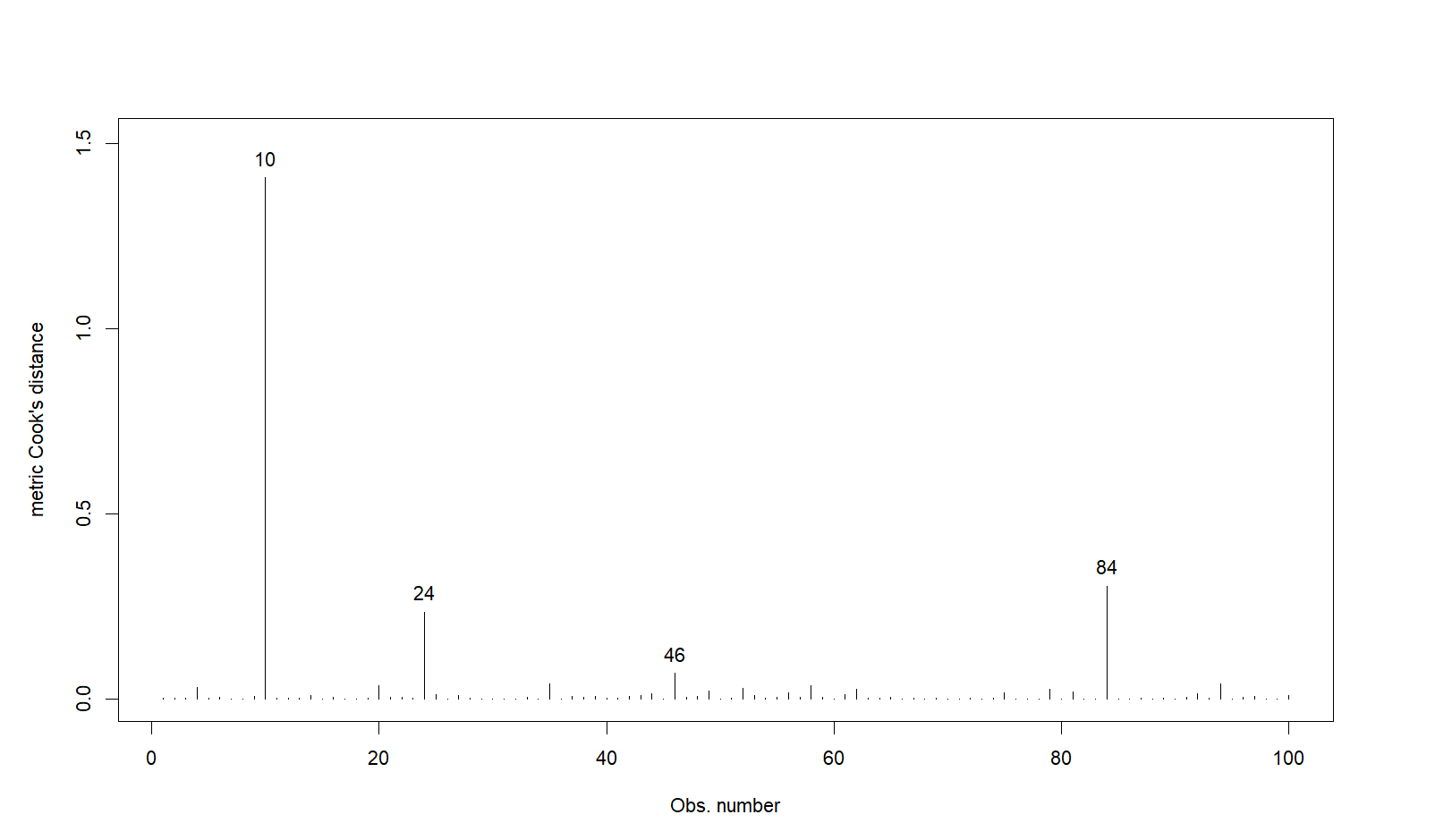}
    \caption{Metric Cook's distances for functional response regression}
    \label{fig:MCD_func}
\end{figure}

\noindent The plot of the functional responses vs the sufficient predictor is given in Figure \ref{fig:func_resp}. The functions that have different periodic spikes and those associated with extremely large values of $\beta^\top\x$ appear to be the ones flagged as influential. 

\begin{figure}[htb!]
    \centering
    \includegraphics[width=0.75\textwidth, height=0.52\textwidth]{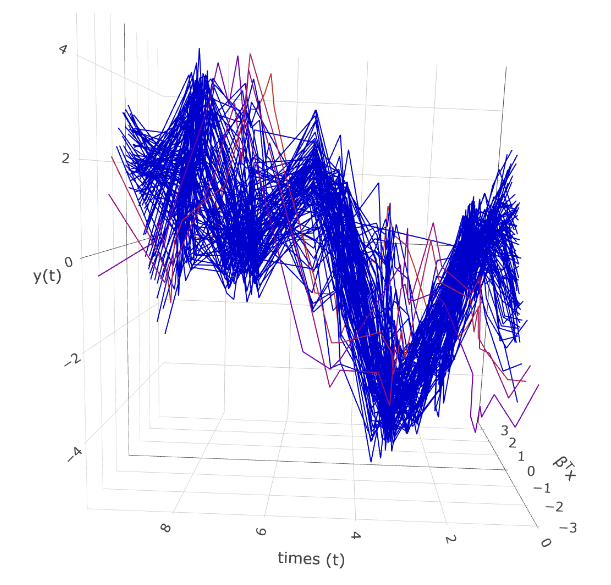}
    \caption{3D plot of functional responses vs sufficient predictor $\beta^\top\x$. The red lines indicate the influential observations based on the metric Cook's distance.}
    \label{fig:func_resp}
\end{figure}

The estimates of the sa-OLS before and after trimming the influential observations are provided in Table \ref{ex:func}.
We see that trimming the influential observations improves the sa-OLS estimates. FOLS estimates were not included here for the same reason as before.

\begin{table}[htb!]
\centering
\caption{* denotes estimates without influential observations}
\label{ex:func}
\begin{tabular}{lccccc|c}
\hline
 & $x_1$ & $x_2$ & $x_3$ & $x_4$ & $x_5$ & $\Delta$ \\ \hline    
sa-OLS  & 0.7533 & 0.5970 & -0.0687 & 0.0197 & 0.0470 & 0.2047 \\
sa-OLS* & 0.7866 & 0.7817 &  0.0251 & 0.0293 & -0.0506 & 0.0812\\
\hline
\end{tabular}
\end{table}

\noindent Again, it is important to note that the anomaly in the functional response could be due to the phase or amplitude variations. Thus, if the investigator chooses to use a metric different from the discrete Fourier distance, we advise that that metric be one that takes into account both the phase and amplitude variations of the input functions.

\subsection{Effects of trimming influential observations}
This part of the experimental studies will demonstrate the impact of influential observations on the central mean space estimates beyond a single sample. To show this, we repeat Examples 1--4 for 500 random samples and report the average value of $\Delta$. We expand the experiments to include two different models for each response space. The simulation experiments will also include estimates for $(n, p) = (200, 10)$ in addition to the $(100, 5)$ already presented. 

For all models, the predictor is generated as $\X \sim t_{20}(\bm 0, \I_p)$ as before and the responses from the four metric spaces are considered. 

{\bf Euclidean response:}
\begin{enumerate}
    \item[(I)] $\Y_i = \sin(\pi/2 + \B^\top\X_i) + 0.5\bm\epsilon_i, \ i=1,\ldots,n$
    \item[(II)] $\Y_i = 0.8(\B^\top\X_i)^3 + 0.5\bm\epsilon_i, \ i=1,\ldots,n$
\end{enumerate}   

{\bf Distribution as response:}
\begin{enumerate}
    \item[(III)]  $Y_i \in \R^{100} \overset{i.i.d.}{\sim} 0.6N(\beta^\top\X_i, 1) + 0.4N(-\beta^\top\X_i, 2)$, \text{ for } $i= 1,\ldots,n.$
    \item[(IV)]  $Y_i \in \R^{100} \overset{i.i.d.}{\sim} Poisson(e^{\beta^\top\X_i})$, \text{ for } $i= 1,\ldots,n.$
\end{enumerate}

{\bf Network as response:}
\begin{enumerate}
    \item[(V)] $Y_i$ follows an Erd\"os–R\'enyi model with probability 
    $p = plogis\big(\sin(\beta^\top\X_i)\big)$, \\
    where $plogis(x) = \cfrac{1}{1+e^{-x}}$.

    \item[(VI)]  $Y_i$ follows a stochastic block model with Poisson distributed edge weights. The block proportions are $(0.4,0.3,0.3)$ with the parameter matrix $\M = 15\I_3 + \lceil{e^{\beta^\top\X_i}} \rceil$. See \cite{leger2016blockmodels} for details.
\end{enumerate}

{\bf Functional response:}
\begin{enumerate}
    \item[(VII)] $Y_i(t) = \alpha_0(t) + 2\sin\left(\pi t/2 + \beta_1^\top\X_{i}\right) + 0.5\epsilon_{i}(t)$,
    \item[(VIII)] $Y_i(t) = \alpha_0(t) + \sin\left(\pi t/2 + \beta^\top\X_{i}\right) + \sin\left(\pi t/2 - \beta^\top\X_{i}\right)  + 0.5\epsilon_{i}(t)$.
\end{enumerate}

\noindent The results of the simulation study are provided in Table \ref{tab:sims}. 

\begin{center}
\small
\begin{longtable}{ccclcc}
\caption{Mean (standard deviation) of $\Delta$ based on 500 random samples.}
\label{tab:sims} \\
\hline
\multicolumn{1}{l}{Response type} & \multicolumn{1}{l}{Model} & \multicolumn{1}{l}{(n, p)} & Method & \multicolumn{1}{l}{All obs.} & \multicolumn{1}{l}{Without influential obs.} \\ \hline
\endhead
\multicolumn{6}{r}{{Continued on next page}} \\ \hline
\endfoot

\hline
\endlastfoot

\multirow{4}{*}{Euclidean} & \multirow{2}{*}{I} & \multirow{1}{*}{(100, 5)} & sa-OLS & 1.1880 (0.0111) & 1.1420 (0.0129) \\
&  & \multirow{1}{*}{(200, 10)} & sa-OLS & 1.3094 (0.0058) & 1.2883 (0.0069) \\
\cline{2-6}  
& \multirow{2}{*}{II} & \multirow{1}{*}{(100, 5)} & sa-OLS & 0.4854 (0.0128) & 0.2590 (0.0067)\\
& & \multirow{1}{*}{(200, 10)} & sa-OLS & 0.6660 (0.0147) & 0.2737 (0.0076)\\
\hline

\multirow{8}{*}{Distribution} & \multirow{4}{*}{III} & \multirow{2}{*}{(100, 5)} & FOLS & 0.1881 (0.0057) & 0.1346 (0.0048) \\
& & & sa-OLS & 0.2009 (0.0055) & 0.1370 (0.0046) \\
& & \multirow{2}{*}{(200, 10)} & FOLS & 0.1972 (0.0049) & 0.1403 (0.0042)\\
& & & sa-OLS & 0.2179 (0.0047) & 0.1462 (0.0043) \\
\cline{2-6} 
& \multirow{4}{*}{IV} & \multirow{2}{*}{(100, 5)} & FOLS & 0.4249 (0.0089) & 0.2123 (0.0051)\\
& & & sa-OLS &  0.2943 (0.0061) & 0.1782 (0.0042)\\
& & \multirow{2}{*}{(200, 10)} & FOLS & 0.5824 (0.0094) & 0.2714 (0.0041)  \\
& & & sa-OLS & 0.3726 (0.0065) & 0.2129 (0.0029) \\
\hline

\multirow{8}{*}{Network} & \multirow{4}{*}{V} & \multirow{2}{*}{(100, 5)} & \multicolumn{1}{c}{sa-OLS (cd)} & 0.1747 (0.0039) & 0.1056 (0.0019) \\
& & & \multicolumn{1}{c}{sa-OLS (dd)} & 0.1865 (0.0040) & 0.1200 (0.0023) \\
& & \multirow{2}{*}{(200, 10)} & \multicolumn{1}{c}{sa-OLS (cd)} & 0.1890 (0.0029) & 0.1184 (0.0029) \\
& & & \multicolumn{1}{c}{sa-OLS (dd)} & 0.2064 (0.0029) & 0.1362 (0.0029)\\
\cline{2-6} 
& \multirow{4}{*}{VI} & \multirow{2}{*}{(100, 5)} & \multicolumn{1}{c}{sa-OLS (cd)} & 0.3926 (0.0071) & 0.3295 (0.0068) \\
& & & \multicolumn{1}{c}{sa-OLS (dd)} & 0.2271 (0.0038) & 0.1989 (0.0041)\\

& & \multirow{2}{*}{(200, 10)} & \multicolumn{1}{c}{sa-OLS (cd)} & 0.4745 (0.0068) & 0.3640 (0.0054) \\
& & & \multicolumn{1}{c}{sa-OLS (dd)} & 0.2479 (0.0029) & 0.2225 (0.0042) \\ 
\hline

\multirow{4}{*}{Functional} & \multicolumn{1}{l}{\multirow{2}{*}{VII}} & \multicolumn{1}{l}{\multirow{1}{*}{(100, 5)}}  & sa-OLS & 0.1785 (0.0039) & 0.0994 (0.0037) \\
& \multicolumn{1}{l}{} & \multicolumn{1}{l}{\multirow{1}{*}{(200, 10)}} & sa-OLS & 0.2039 (0.0031) & 0.1244 (0.0045) \\
\cline{2-6} 
& \multicolumn{1}{l}{\multirow{2}{*}{VIII}} & \multicolumn{1}{l}{\multirow{1}{*}{(100, 5)}}  & sa-OLS & 1.0567 (0.0144) & 0.9142 (0.0160) \\
& \multicolumn{1}{l}{} & \multicolumn{1}{l}{\multirow{1}{*}{(200, 10)}} & sa-OLS & 1.2231 (0.0090) & 1.1161 (0.0122) \\
\hline
\end{longtable}
\end{center}

Across models and response spaces, trimming the influential observations improved the estimation accuracy of almost all the methods. Moreover, the influential observations appear to have more impact in models with unbounded support. For example, for Euclidean responses, we see drastic changes in estimation accuracy in model II than in model I. This is because the $\sin(.)$ function in model I bounds the mean component between -1 and 1, thus, curtailing the severe impact of the outliers. 

\section{Applications to Real Data \label{sec:4}}
We now turn our attention to two real applications to see how the metric Cook's distance fares in practice. The data sets considered are the same ones used in \cite{soale2023data}. 

\subsection{COVID-19 transmission in the State of Texas}
In this application, the responses are the distributions of total new COVID-19 cases per 100,000 persons in the last 7 days reported at the county level between 08/1/2021 and 02/21/2022 for the State of Texas, which consists of 254 counties. The predictors are nine demographic characteristics of the counties from the 2020 American Community Survey. Both data sets are publicly available at \cite{coviddata} and \cite{acs2020}, respectively. Also, see \cite{soale2023data} for detailed data description.

The influential observations based on the metric Cook's distances for the regression between the COVID-19 transmission distributions vs the demographic characteristics are given in Figure \ref{fig:MCD_covid}.
\begin{figure}[htb!]
    \centering
    \includegraphics[width=0.9\textwidth, height=0.4\textwidth]{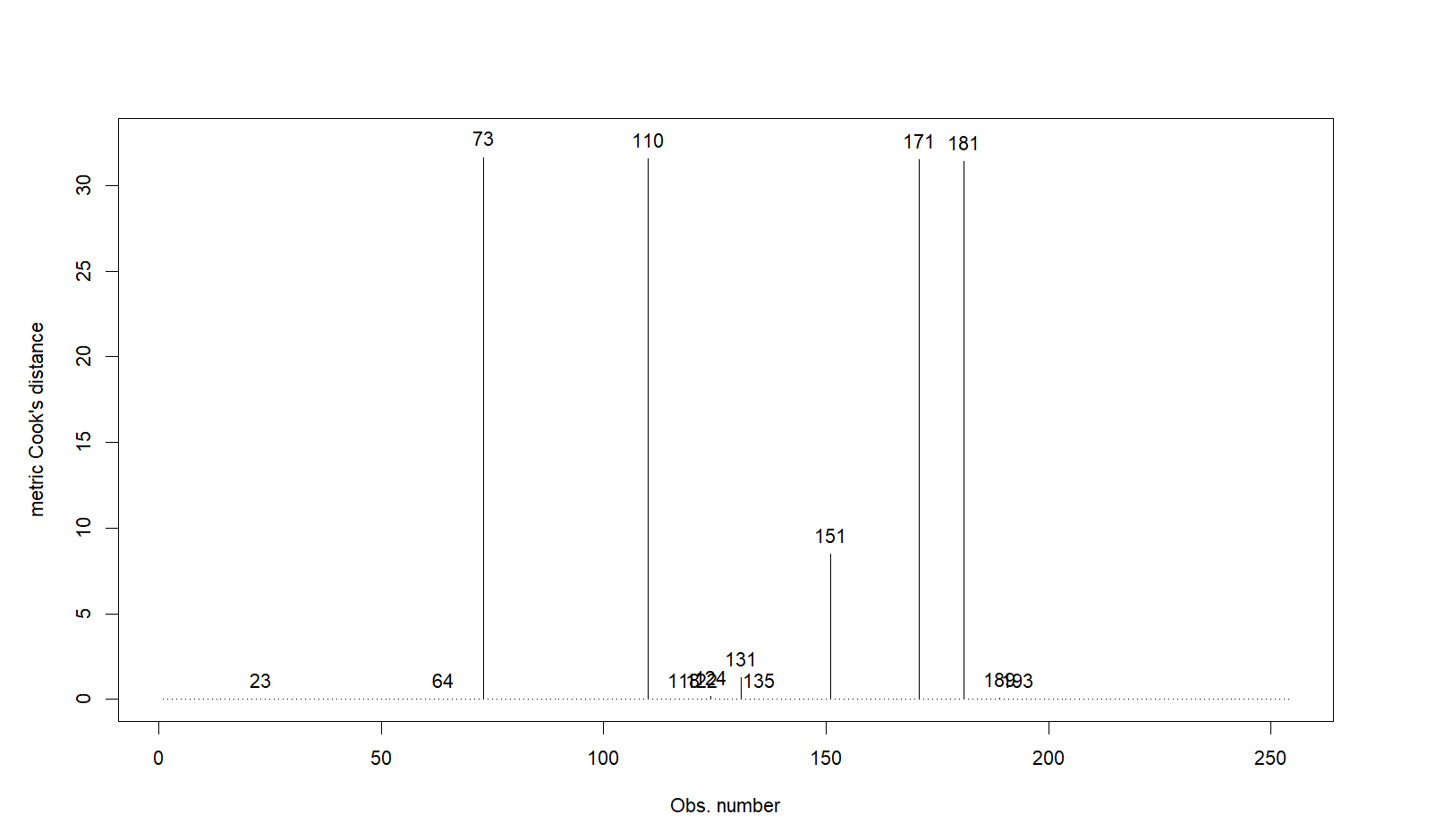}
    \caption{Influential observations based on metric Cook's distance for the regression between COVID-19 transmission distributions vs demographic characteristics. }
    \label{fig:MCD_covid}
\end{figure}
The counties corresponding to the influential observations: 73, 110, 151, 171, and 181 are Falls County, Hockley County, Loving County, Moore County, and Orange County. 

\begin{table}[htb!]
\centering
\caption{estimated bases of the central mean space. * denotes estimates after omitting influential observations.}
\label{tbl:covid}
\begin{tabular}{lcccc}
\hline
                             & FOLS  & FOLS* & sa-OLS  & sa-OLS* \\ \hline
\% Non-Hispanic Blacks       & 0.0612  & 0.1860   & 0.1728  & 0.3102   \\
\% Hispanics                 & -0.2771 & -0.0339 & -0.3974 & -0.3431 \\
\% Adults 65+                & -0.7254 & -0.5820  & -0.8304 & -0.8659 \\
\% No high school diploma    & 0.3848  & -0.0789 & 0.3949  & -0.1724 \\
\% Living below poverty line & 0.2294  & -0.5378 & 0.0284  & -0.1615 \\
\% Unemployed                & -0.0225 & 0.3190   & 0.0455  & 0.2875 \\
\% Renter-occupied homes     & -0.8263 & 0.6574  & -0.6710  & 0.0731 \\
\% On public assistance      & 0.0359  & -0.1459 & 0.0969  & 0.1479 \\ \hline
\end{tabular}
\end{table}

The estimated FOLS and sa-OLS bases with and without the influential observations are provided in Table \ref{tbl:covid}. In the Table, we see significant changes in some of the coefficients after trimming the influential observations. As the true basis for the $\cms$ is unknown, we cannot estimate $\Delta$. Therefore, we provide some visualizations to illustrate the effects of trimming the influential observations in the estimated sufficient summary plot for some measures of the center and spread of the distributions based on the sa-OLS direction.

In Figure \ref{fig:summaryplot_covid}, we notice that trimming influential observations from the data significantly changes the slopes of the linear model fit, especially for the mean and standard deviations of the distributions. This is not surprising as the mean and standard deviations are known to be susceptible to outliers. The median and mode do not appear to be affected much by the presence of the influential observations.

\begin{figure}[htb!]
    \centering
    \includegraphics[width=\textwidth,height=3.2in]{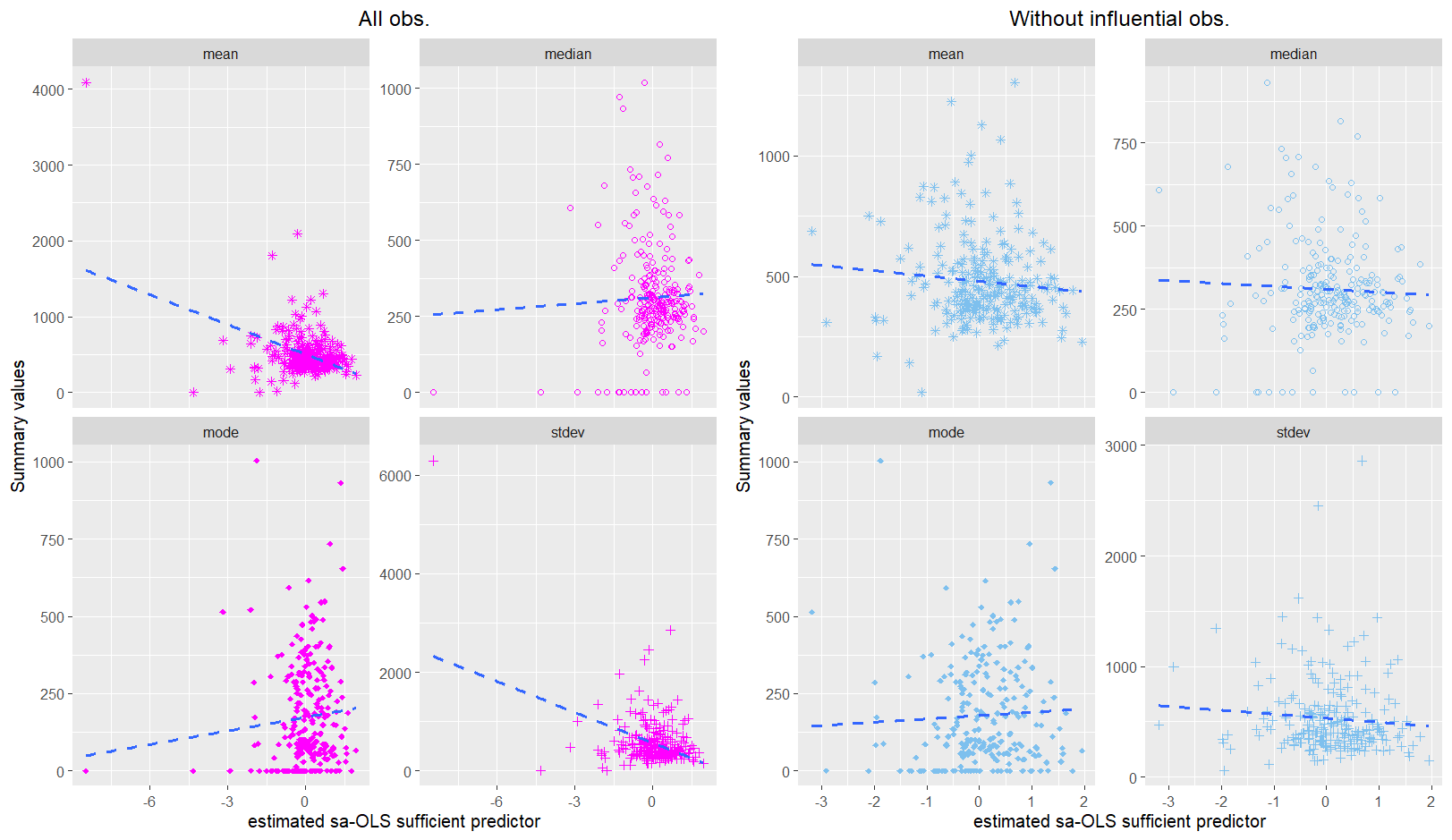}
    \caption{Summary plots for the regression between COVID-19 transmission distributions vs demographic characteristics with and without the influential observations. The blue dash lines indicate the fitted lines for the simple linear regression.}
    \label{fig:summaryplot_covid}
\end{figure}

\subsection{Human brain structural connectivity}
For this application, the weighted connectivity matrices among 90 cortical regions of interests in the brain of the study subjects are taken as the responses. The sample consists of 88 healthy individuals whose age, weight, and height are taken as the predictors. This data is also publicly available on the Open Science Framework (OSF) at \href{https://osf.io/yw5vf/}{https://osf.io/yw5vf/} and published in \cite{vskoch2022human}. 

The influential observations for the regression between the structural brain connectivity networks vs the age, height, and weight of subjects are given in Figure \ref{fig:MCD_brain}.

\begin{figure}[htb!]
    \centering
    \includegraphics[width=0.85\textwidth,height=0.4\textwidth]{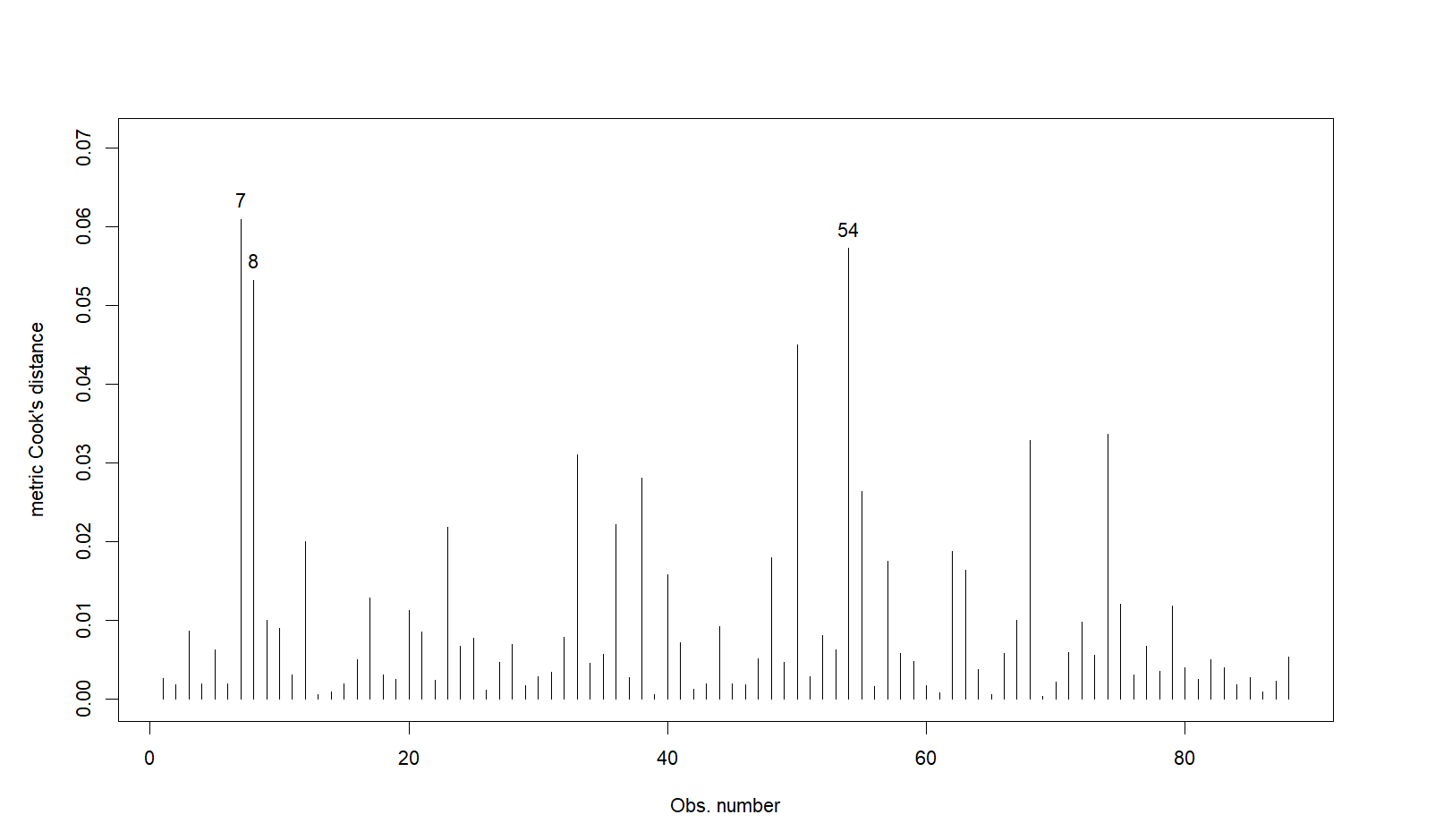}
    \caption{Influential observations based on metric Cook's distance for the regression between structural brain connectivity networks vs age, height, and weight of subjects. }
    \label{fig:MCD_brain}
\end{figure}
\noindent Based on the metric Cook's distances, Subjects 7, 8, and 54 have brain networks that differ significantly from the most of the subjects in the study. The respective ages of these outliers are 48.7, 37.7, and 45.3 years with corresponding weights 67, 62, and 60 kg. As demonstrated in Figure \ref{fig:net3}, it is hard to tell that a given brain network has an anomalous structural connectivity, especially when the number of links in the network is high. That notwithstanding, removing the influential observations leads to a significant change in the coefficient for age across all the estimators given in Table \ref{tbl:brain}, which warrants further investigation.

\begin{table}[htb!]
\small
\centering
\caption{estimated bases of the central mean space. * denotes estimates after omiting influential observations.}
\label{tbl:brain}
\begin{tabular}{lcccc}
\hline
& \multicolumn{1}{l}{sa-OLS (CD)} & \multicolumn{1}{l}{sa-OLS (CD)*} & \multicolumn{1}{l}{sa-OLS (DD)} & \multicolumn{1}{l}{sa-OLS (DD)*} \\ \hline
Age & -0.0805 & -0.1258  & -0.2234 & -0.4726 \\
Weight & -1.3446 & -1.3182 & -1.2794 & -1.1156 \\
Height & 0.6452 & 0.6116 & 1.2263 & 1.1323 \\ \hline
\end{tabular}
\end{table}

The summary plots in Figure \ref{fig:summaryplot_brain} show the relationship between some network properties and age. 

\begin{figure}[htb!]
    \centering
    \includegraphics[width=\textwidth,height=0.5\textwidth]{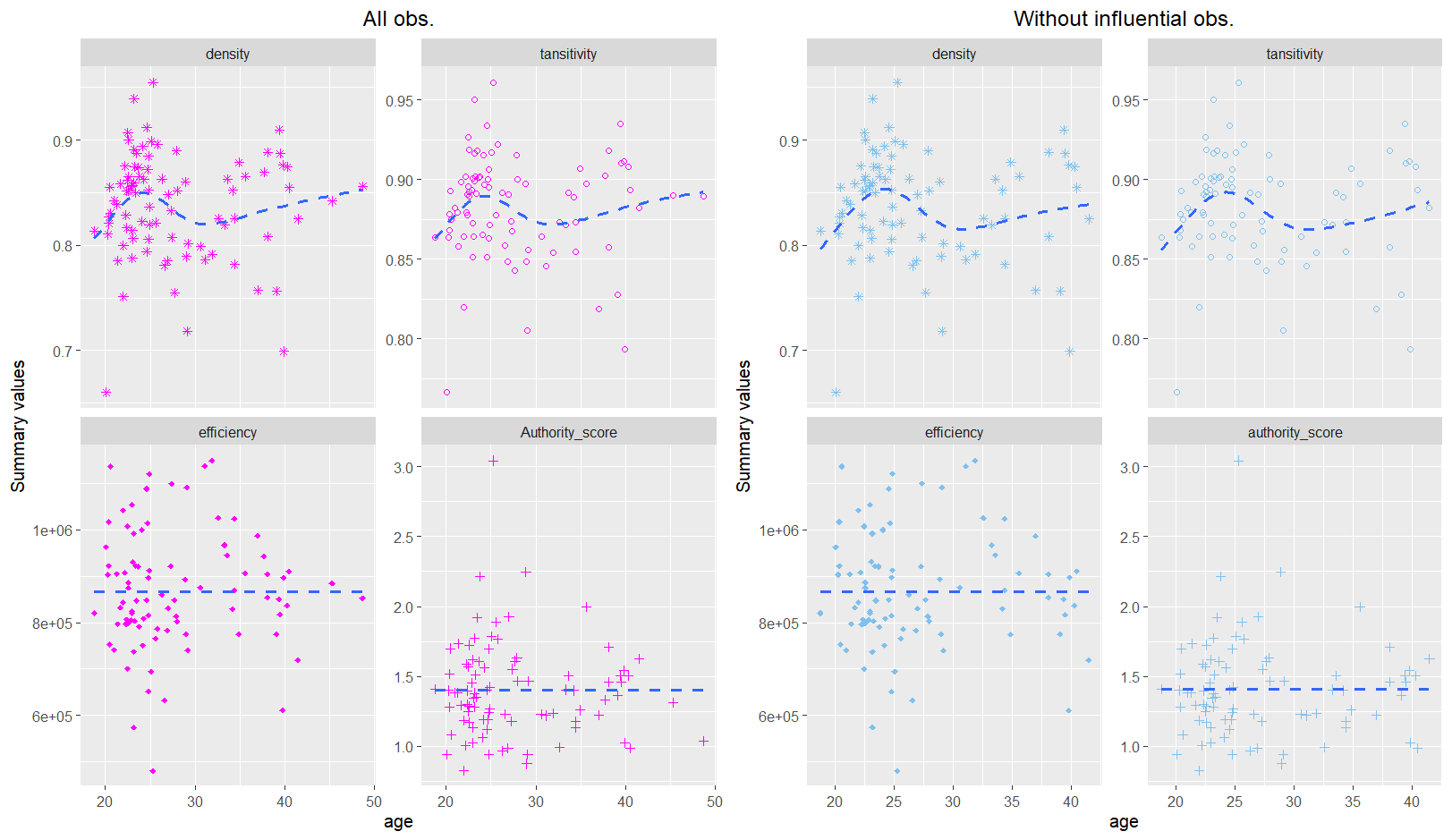}
    \caption{Summary plots for the regression between some properties of the structural brain connectivity networks vs age of subjects with and without the influential observations. The blue dash lines represent the LOESS curve.}
    \label{fig:summaryplot_brain}
\end{figure}
\noindent The pattern of the LOESS curve does not appear to vary much after trimming the influential observations. However, we notice the exclusion of the two oldest subjects, resulting in a change in limit of the horizontal axis after trimming the influential observations. 

\section{Conclusion\label{sec:5}}
In this paper, we presented a metric Cook's distance technique for detecting influential observations for regression between metric-valued responses and Euclidean predictors. As demonstrated in Example 1, even in the Euclidean space, the proposed metric Cook's distance can be very useful in regression with rank-reduced vector responses. Also, as demonstrated, influential observations can have a significant impact on the estimate of the central mean space even in the emerging random object regression. In a nutshell, the metric Cook's distance is an essential tool to aid investigators detect such unusual observations or anomalous behavior, especially when they are more subtle. 

We want to stress that rather than use a hard threshold for classifying an observation as influential, we advise investigators to examine observations with high Cook's distance values closely. Another caveat we want to point out is that trimming the influential observations may not necessarily be the best way to handle influential observations. Instead, the investigator should examine and consider the potential reasons for such occurrence. If mean regression is not the main objective, the alternative could be to estimate the central space, which encompasses the central mean space as well as dependence in higher moments. The central space estimates tend to be more robust to outliers compared to the central mean space estimates.

Moreover, while we did not investigate the impact of metric choice in this paper, we advise that investigators carefully choose metrics that capture the essential properties of the objects of interest. Finally, a future study that explores the effects of outliers beyond the single index model and in regression between non-Euclidean objects will be an interesting to extension.







\bibliographystyle{apalike}

\bibliography{ref}

\section*{Appendix}
\begin{proof}[Proof of Lemma \ref{lemma:surr_sdr}]\quad \\
If $Y \indep \X | \beta^\top\X$, then $(Y, \tY) \indep \X | \beta^\top\X$ also hold, for any random copy $\tY$ of $Y$. By Theorem 2.3 of Li(2018), we have that $\phi(Y, \tY) \indep \X | \beta^\top\X$ \cite{li2018sufficient}.
\end{proof}

\begin{proof}[Proof of Theorem \ref{thrm:surr_ols}]\quad \\
By Lemma \ref{lemma:surr_sdr}, $\mS_{S^Y|\X} \subseteq \spc$. The ordinary least squares estimator  $\Sig^{-1}\Sig_{X S^Y} \in \mS_{\E(S^Y|\X)} \subseteq \mS_{S^Y|\X}$ by Theorem 8.3 of \cite{li2018sufficient}, which leads to the desired results.  
\end{proof}

\end{document}